\documentclass[a4paper,11pt]{article}
\usepackage{fullpage}
\usepackage{microtype} 
\usepackage{pdfpages}
\usepackage[T1]{fontenc}
\usepackage{mathtools,amsmath,xspace}
\usepackage{amssymb,amsthm}
\usepackage{graphicx}
\usepackage{hyperref} 
\usepackage[nodayofweek]{datetime} 
\usepackage{xspace}
\usepackage{soul}
\usepackage[]{algorithm2e}
\newtheorem{theorem}{Theorem} 
\newtheorem{lemma}[theorem]{Lemma}

\newbox\ProofSym \setbox\ProofSym=\hbox{%
	\unitlength=0.18ex%
	\begin{picture}(10,10) \put(0,0){\framebox(9,9){}}
	\put(0,3){\framebox(6,6){}}
	\end{picture}}

\let\geq\geqslant \let\leq\leqslant

\newcommand{\conv}{\mathsf{CH}}
\newcommand{\bbst}[1]{\ensuremath{\mathcal{T}(#1)}}

\title{Efficient Planar Two-Center Algorithms\thanks{This work was supported by Institute of Information \& communications Technology Planning \& Evaluation (IITP) grant funded by the Korea government (MSIT) (IITP-2017-0-00905, Software Star Lab
(Optimal Data Structure and Algorithmic Applications in Dynamic Geometric Environment).}%optional
}

\author{Jongmin Choi\thanks{Department of Computer Science and Engineering, Pohang University of Science and Technology, Pohang, Korea, \texttt{icothos@postech.ac.kr}}
\and Hee-Kap Ahn\thanks{Department of Computer Science and Engineering, Graduate School of Artificial Intelligence, Pohang University of Science and Technology, Pohang, Korea, \texttt{heekap@postech.ac.kr} \texttt{https://orcid.org/0000-0001-7177-1679}}}

\begin{document}
\date{}
\maketitle
\begin{abstract}
We consider the planar Euclidean two-center problem in which given $n$ points in the plane
we are to find two congruent disks of the smallest radius covering the points. 
We present a deterministic $O(n \log n)$-time algorithm for the case that the centers of 
the two optimal disks are close to each other, that is, the overlap of the two optimal disks 
is a constant fraction of the disk area. 
We also present a deterministic $O(n\log n)$-time algorithm 
for the case that the input points are in convex position. Both results improve the previous best $O(n\log n\log\log n)$ bound on the problems. 
\end{abstract}

\section{Introduction}
In the planar 2-center problem we are given a set $S$ of $n$ points in the
plane and want to find two congruent disks of the smallest radius covering the
points. This is a special case (with 2 supply points) of the facility location problem in
which given a set of demands in $\mathbb{R}^d$ we are to find a set of
supply points in $\mathbb{R}^d$ minimizing the maximum distance from a
demand point to its closest supply point.

There has been a fair amount of work on the planar 2-center problem in
the 1990s.  Hershberger and Suri~\cite{HERSHBERGER1996453} considered
a decision version of the problem: given a radius $r$, determine if
$S$ can be covered by two disks of radius $r$.  They gave an
$O(n^2\log n)$-time algorithm for the problem.  This was improved
slightly by Hershberger to $O(n^2)$ time~\cite{HERSHBERGER199323}. 
By combining this result with the parametric-search paradigm of
Megiddo~\cite{Megiddo1983}, Agarwal and Sharir~\cite{Agarwal1994} gave
an $O(n^2\log^3 n)$-time algorithm for the planar 2-center problem. A few
other results include the expander-based approach by Katz and
Sharir~\cite{Katz1993} avoiding the parametric search, a randomized
algorithm with expected $O(n^2\log^2n\log\log n)$ time by
Eppstein~\cite{Eppstein1992}, and a deterministic $O(n^2\log n)$-time
algorithm by Jaromczyk and Kowaluk~\cite{Jaromczyk1994} using new
geometric insights.

% \complain{Sharir divide case as $d<r^*$ ,$r^*<d<3r^*$ and $3r^*<d$ but, 
% Eppstein divide case as $d<2r^*$ ,$\epsilon r^*<d$. Writed in bellow paragraphs about Eppstein not Sharir.}
In 1997, Sharir made a 
% Sharir~\cite{Sharir1997} made a 
substantial breakthrough by 
dividing the problem into cases, depending on the optimal radius $r^*$ and
 the distance $d$  %$d<r^*, r^*<d<3r^*,$ and $3r^*<d$, where
between the optimal centers. He presented a $O(n\log^9 n)$-time algorithm
by combining an $O(n\log^2 n)$-time algorithm
deciding for a given radius $r$ if the points can be covered by two
disks of radius $r$ and parametric search.
Later, Eppstein~\cite{Eppstein1997} considered the case
that % the two optimal disks are disjoint, 
% more precisely,
$d\geq cr^*$ for any positive constant $c$ %and improved the result 
% by considering two cases, $d<2r^*$ and $\epsilon r^*<d$,
and presented an $O(n\log^2 n)$-time algorithm.

%\ccheck{For the case $d\geq 2r^*$, Sharir~\cite{Sharir1997} 
%presented a method to find $r^*$ in $O(n\polylog n)$ time
%% depending on if the optimal disks are well-separated or not.  
%by combining an $O(n\log^2 n)$-time algorithm
%deciding for a given radius $r$ if the points can be covered by two
%disks of radius $r$ and parametric search. % Sharir  
%
%\complain{Check if this is correct! $O(n\log^9 n)$ time for Sharir, $O(n\log^9 n)$? More specific?}

%For a fixed constant \ccheck{$c\in (1,2)$,} \was{$c\in (0,2)$} consider the case that the distance
%between the optimal centers is at least $cr^*$, where $r^*$ is the
%optimal radius. 
%From now on, we will \ccheck{call this case \emph{the
%  well-separated case}.} See Figure~\ref{fig:restricted} left for an illustration.  
  
For the case that the distance between the two optimal centers is % smaller than 
$cr^*$ for a fixed constant $c\in [0,2)$,
the overlap of the optimal disks is a constant
fraction of the disk area. It is known that one can generate a constant number of
points in linear time such that at least one of them lies in the overlap of the optimal
disks~\cite{Sharir1997,Eppstein1997}. 
Thus, the problem reduces to the following restricted problem. % \medskip
See Figure~\ref{fig:restricted} for an illustration.
\vspace{-.5em}
\begin{quote}
  % \noindent
  \textbf{The restricted 2-center problem.} Given a set $S$ of $n$ points
  and a point $o$ in the plane, find two congruent disks $D_1,D_2$ of
  the smallest radius such that $o\in D_1\cap D_2$ and 
  $S$ is covered by $D_1$ and $D_2$. % \medskip
\end{quote}
\vspace{-.5em}
\begin{figure}[tb]
	\centering
	\includegraphics[width=.3\textwidth]{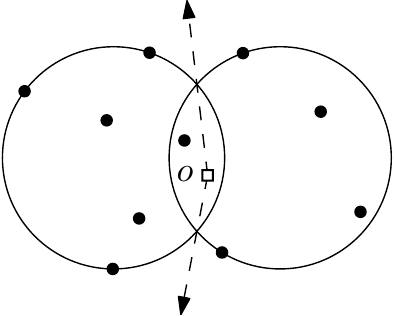}
	\caption{The restricted 2-center problem: two congruent disks of the smallest radius enclosing the points overlap.} 
%	(a) The distance $d$ between the optimal centers is larger than $2r^*$. (b)  }
	\label{fig:restricted}
\end{figure}

\noindent
The restricted 2-center problem can be solved by Sharir's algorithm in
% considered the  and gave an
$O(n\log^9 n)$ time~\cite{Sharir1997}.
Eppstein gave a randomized algorithm which solves the restricted
2-center problem in $O(n\log n\log\log n)$ expected
time~\cite{Eppstein1997}.

In 1999, Chan~\cite{CHAN1999189} made an improvement to
Eppstein's algorithm by showing that the restricted 2-center problem 
can be solved in $O(n\log n)$ time with high probability. 
He also presented a deterministic 
$O(n\log^2 n\log^2\log n)$-time algorithm for the restricted 2-center problem 
by using parametric search.  Together with the $O(n\log^2 n)$-time algorithm for the
well-separated case by Eppstein~\cite{Eppstein1997}, the planar
2-center problem is solved in $O(n\log^2 n)$ time with high
probability, or in $O(n\log^2 n\log^2\log n)$ deterministic time.

Since then there was no improvement over 20 years until
Wang~\cite{WANG2020} very recently gave an $O(n\log n\log\log n)$-time
deterministic algorithm for the restricted 2-center problem, improving
on Chan's method by a $\log n\log\log n$ factor.  Combining Eppstein's
$O(n\log^2 n)$-time algorithm for the well-separated case, the planar
2-center problem can be solved in $O(n\log^2 n)$ time
deterministically.
% \footnote{While we were preparing this submission,
%  we got to know that Wang made this improvement to the restricted 2-center
%  problem which was accepted to
%  SoCG 2020. He kindly provided us the paper upon our request.}  
The improvement is from a new $O(n)$-time sequential decision algorithm, 
precomputing the intersections of disks that are frequently used 
in the parallel algorithms, and a dynamic data structure called \emph{the circular hull}
that maintains the intersection of
the disks containing a fixed set of points. See the right figure in~Figure~\ref{fig:intersection-chull}.

Eppstein showed that 
any deterministic algorithm for the two-center problem requires 
$\Omega(n\log n)$ time in the algebraic decision tree model of computation 
by reduction from max gap problem~\cite{Eppstein1997}. So there is
some gap between the best known running time and the lower bound.
% \complain{Lower bound for both the problem ? Add it here!}

% \complain{A bit more on this result.}
% the circular hull of a set of points.

% \begin{itemize}
% \item \complain{A few wrong results?}
% \item \complain{Previous work on points in convex position}
% \end{itemize}
% There also have been a few works on the special case that points are
% in convex position.

\subsection{Our results}
We consider the restricted 2-center problem and show
that the problem can be solved in $O(n\log n)$ time
deterministically. This improves the previous best
$O(n\log n\log\log n)$ bound by Wang by a $\log\log n$
factor. 
% \ccheck{and matches the lower bound.}  \complain{a brief
%  introduction to the improvements.}

% Summary of our algorithm
Our algorithm follows the framework by Wang~\cite{WANG2020} (and thus
by Chan~\cite{CHAN1999189}) and uses the $O(n)$-time sequential
decision algorithm by Wang. But for the parallel algorithm, 
we use a different approach of running a decision algorithm in $O(\log n)$ 
parallel steps using $O(n)$ processors after $O(n \log n)$-time preprocessing.

%Our algorithm partitions $P$ into two subsets % $S^+$ and $S^-$, 
%one lying above the $x$-axis and one lying below the $x$-axis, in a coordinate
%system sorted around $o$ in counterclockwise order.
% , we can give them
% indices in the order.
By following Wang's approach, our algorithm divides the points into a certain number of groups,
computes for each group the common intersection of the disks centered at 
some points in the group and applies binary search on the intervals of indices independently.
Wang used $O(\log n)$ processors to compute the intersections in $O(\log n \log \log n)$ time.
In our algorithm, we increase the number of processors to $O(\log^2 n)$ so that 
the time complexity to compute the intersections decreases to $O(\log n)$.

Another difference to the approach by Wang is that 
for each group, we additionally construct a data structure that 
determines the emptiness of intersections with respect to some common intersection.
%\JMremove{the emptiness of the common intersection of the disks.} \complain{???}
This can be done in $O(\log n)$ time using $O(\log^6 n)$ processors. 
We use this data structure in binary search steps. 
It takes $O(\log^3\log n)$ time to determine if the common intersection is empty or not,
which improves the $O(\log n)$-time result by Wang.

Thus our decision algorithm takes $O(\log n)$ time using $O(n)$ processors,
after $O(n\log n)$-time preprocessing. We apply Cole's parametric 
search technique~\cite{Cole-parametric} to compute $r^*$ in $O((T_S+Q)(T_P+ \log Q))$ time, 
where $T_S=O(n)$ is sequential decision time, 
$T_P=O(\log n)$ is parallel decision time and $Q=O(n)$ 
is number of processor for parallel decision algorithm. 
Thus the optimal radius $r^*$ for
the restricted 2-center can be computed in $O(n\log n)$ time.
% From this, our algorithm computes an optimal solution for
% the restricted 2-center in $O(n\log n)$ time.} \complain{More details on all these complexities?} 

We also consider the 2-center problem for the special case that the input points 
are in convex position. Kim and Shin~\cite{kim2000} considered a variant 
of this problem for convex polygons:
Given a convex polygon, find two smallest congruent 
disks whose union contains the polygon.
They presented an $O(n\log^2 n)$-time algorithm for a convex $n$-gon
and 
% that finds two congruent smallest disks
% covering a convex $n$-gon. They 
claimed that the 2-center problem for $n$ points in
convex position can be solved similarly in the same time. But Tan~\cite{Tan2017}
pointed out a mistake in the time analysis, and presented an $O(n\log^2 n)$-time algorithm for the problem. 
But Wang~\cite{WANG2020} showed a counterexample to the algorithm by Tan, 
and gave an $O(n\log n\log\log n)$-time algorithm for this problem~\cite{WANG2020}.

We present a deterministic $O(n\log n)$-time algorithm 
that computes the optimal 2-center for the points in convex position. 
We first spend $O(n\log n)$ time 
to find a point contained in the overlap of the two optimal disks 
and a line separating the input points into two subsets. %$S^+$ and $S^-$.
Then we apply our algorithm
for the restricted 2-center problem to find an optimal pair of disks
covering the input points.

\section{Preliminaries}
For a point set $X$ in the plane, we denote by $I_r(X)$ the common
intersection of the disks with radius $r$, one centered at each point in
$X$. %  \complain{we need to define circular hull in this moment?} 
The circular hull $\alpha_r(X)$ of $X$ 
% (also known as the \emph{$\alpha$-hull} with $\alpha=1$) 
is the common intersection of all disks of radius $r$ containing $X$. 
See Figure~\ref{fig:intersection-chull} for an illustration.
\begin{figure}[tb]
	\centering
	\includegraphics[scale=0.9]{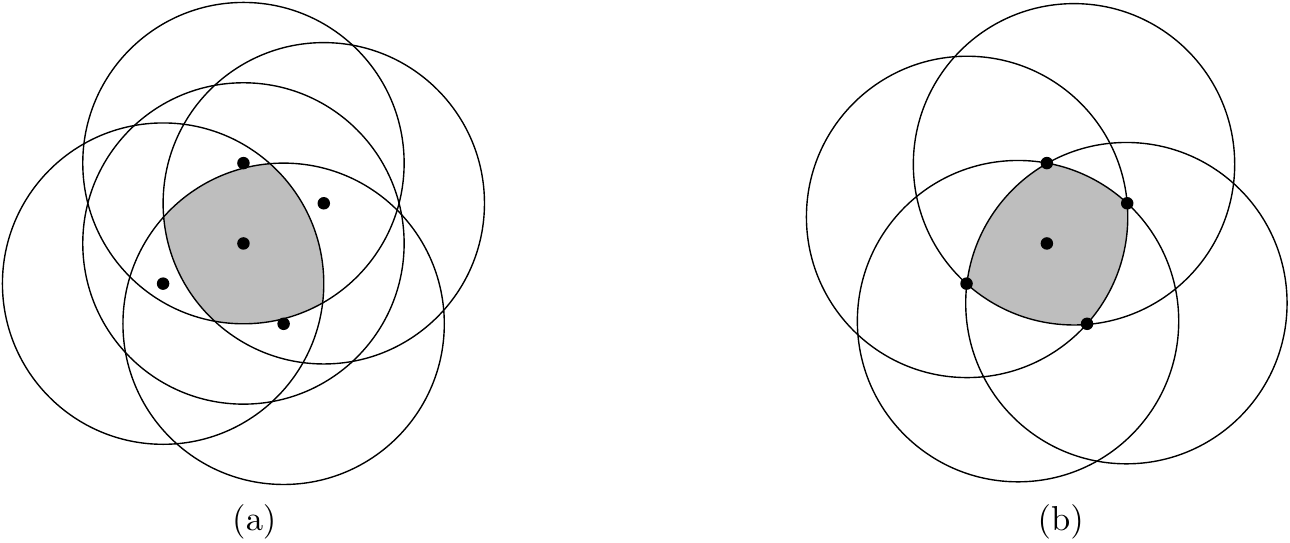}
	\caption{Five points (black dots) in $X$. (a) The intersection $I_r(X)$. (b) The circular hull $\alpha_r(X)$.}
	\label{fig:intersection-chull}
\end{figure}

Clearly, both
$I_r(X)$ and $\alpha_r(X)$ are convex.  Observe that $I_r(X)$ and
$\alpha_r(X)$ are dual to each other in the sense that every arc of
$I_r(X)$ is on the circle of radius $r$ centered at a vertex of
$\alpha_r(X)$ and every arc of $\alpha_r(X)$ is on the circle of
radius $r$ centered at a vertex of $I_r(X)$.  We may simply use $I(X)$
and $\alpha(X)$ instead of $I_r(X)$ and $\alpha_r(X)$ if they are
understood from the context. 

For a compact set $C$ in the plane, we use $\partial C$ to denote
the boundary of $C$. For any two points $p,q\in\partial C$, we use $C[p,q]$ 
to denote the part of the boundary
of $C$ from $p$ to $q$ in counterclockwise order.

Let $S$ be a set of $n$ points in the plane, and let $D^*_1$ and
$D^*_2$ be the two congruent disks of an optimal solution for the
restricted 2-center problem on $S$.
  We can find in $O(n)$ time
a constant number of points such that at least one of them lies in
$D^*_1\cap D^*_2$. We use $o$ to denote a point in $D^*_1\cap D^*_2$.

Without loss of generality, we assume that $o$ is at the origin of the
coordinate system. 
Let $S^+$ denote the set of points of $S$ that lie above the $x$-axis, and
$S^-$ denote $S\setminus S^+$. For ease of description, we assume that
both $S^+$ and $S^-$ have $n$ points. Let
% We assume that both $S^+$ and $S^-$ are not empty.
$s^+_1,s^+_2,\ldots, s^+_n$ be the points of $S^+$ sorted around
  $o$ in counterclockwise order and $s^-_1,s^-_2,\ldots, s^-_n$ be the
  points of $S^-$ sorted around $o$ in counterclockwise order such
  that they appear in order from $s^+_1$ to $s^+_n$ and then $s^-_1$ to
  $s^-_n$ around $o$ in counterclockwise order.  
    For any
  $1\leq i\leq j\leq n$, let $S^+[i,j]=\{s^+_i, s^+_{i+1},\ldots, s^+_j\}$ and
  $S^-[i,j]=\{s^-_i, s^-_{i+1},\ldots, s^-_j\}$.

% For ease of discussion, ...
There are two rays
$\mu_1, \mu_2$ emanating from $o$ that separate the points of $S$ into two
subsets, each covered by one optimal disk. See Figure~\ref{fig:restricted}.
Since such two rays
are always separated by the $x$-axis or the $y$-axis~\cite{CHAN1999189},
we simply assume that 
% We can find a  coordinate system with origin at $o$ such that 
  $\mu_1$ goes upward and $\mu_2$ goes downward. %\complain{time?}
  
For two indices $i,j\in[0,n]$, let $A[i,j]$ denote the radius of the smallest disk enclosing
$S^+[i+1,n] \cup S^-[1,j]$, and let $B[i,j]$ denote the radius of the smallest
disk enclosing $S^+[1,i] \cup S^-[j+1,n]$. For convenience, we let
$S^+[n+1,n]=S^+[1,0]=S^-[n+1,n]=S^-[1,0]=\emptyset$. 
% Note that $A$ and $B$ are defined for indices $i, j \in [0,n]$.  

For two points $p$ and $q$ in the plane, we use 
$\overrightarrow{pq}$ to denote a ray emanating from $p$ going towards $q$.

\section{Improved algorithm for the restricted 2-center problem}
\label{sec:ours}

Our algorithm follows the framework by Wang~\cite{WANG2020} (and thus
by Chan~\cite{CHAN1999189}), and uses the $O(n)$-time sequential
decision algorithm (after $O(n\log n)$-time preprocessing) by Wang.
% after $O(n\log n)$-time preprocessing as Wang's method in
% Section~\ref{sec:Wang.sequential}.
But for the parallel algorithm, 
we use a different approach of running a decision algorithm in $O(\log n)$ 
parallel steps using $O(n)$ processors after $O(n \log n)$-time preprocessing.
% For a parallel algorithm, we give a decision algorithm running in
% $O(\log n)$ parallel steps using $O(n)$ processors after $O(n \log n)$-time preprocessing.

Let $r_{ij} = \max (A[i,j],B[i,j])$ and
% where the minimum radius disk that each disk cover set
% $S^+[i+1,n] \cup S^-[1,j]$ and $S^+[1,i] \cup S^-[j+1,n]$.
$r^*_{i}= \min_{0 \leq j \leq n} r_{ij}$. 
Then
$r^*=\min_{0 \leq i \leq n} r^*_{i}$, where $r^*$ is 
the optimal radius for the restricted 2-center problem on $S$.
Let $R=(r_{ij})$ be the $n\times n$ matrix whose $(i,j)$-entry is $r_{ij}$. 
Let $j(i)=\arg_j\min r_{ij}$.
Observe that $R$ is \emph{monotone} by the definition of $r_{ij}$, that is, 
for any two indices $i_1> i_2$ we have $j(i_1)\geq j(i_2)$.

We search for $r^*$ among the entries in $R$. 
We can determine whether $r^*<r $ by checking if there is some $i$ with $r^*_i < r$.
Observe that for a fixed $i$, $A[i,j]$ increases and $B[i,j]$ decreases as $j$ increases. 
Therefore, $r_{ij}$ decreases and then increases as $j$ increases. 
By this property we can apply binary search in determining for a given $r$ if $r^*_i< r$ or not.
However, it takes too much time for our purpose to apply binary search directly 
on the entries of $R$.
Instead, we restrict binary search to certain elements of $R$ using the monotone property of $R$, 
and construct a few data structures to speed up the decision procedure as follows.

The algorithm consists of three phases.
%Preprocessing independent with $r$, Preprocessing with $r$ and Binary search. 
% In the preprocessing phase independent to $r$,
In Phase 1, our algorithm constructs a data structure on $S$ such that 
given a query consisting of $r>0$ and an interval $[i,j]$ it returns 
$I_r(S^+[i,j])$ or $I_r(S^-[i,j])$. % \complain{a generic structure working for any given $r$?}
It also reduces the search space % the candidates 
in $R$ by evaluating on $r_{ij}$ at every
$(\log^6 n)$-th $j$ values from $j=1$. This gives us a set of $O(n/\log^6 n)$
disjoint submatrices of $R$ with height $\log^6 n$. Then it divides 
each of these submatrices further such that its width is at most $\log^6 n$. 
So there are $O(n/\log^6 n)$ groups. See Figure~\ref{fig:monotone}.
% \complain{Is this what you intend?}
In Phase 2, given $r>0$, our algorithm constructs 
a data structure for each submatrix obtained from Phase 1 such that
given two indices $i$ and $j$ with $r_{ij}$ in the submatrix it determines whether 
$I_r(S^+[i,n])\cap I_r(S^-[1,j])=\emptyset$ or not. 
% that is, the emptiness of the common intersection of disks. 
In Phase 3, our algorithm applies binary search to find a radius smaller than 
$r$ among the elements in each submatrix using the intersection emptiness 
queries on the data structure
in Phase 2. 

\subsection{Phase 1: Preprocessing}
In Phase 1, we construct a data structure on $S$ and reduce the search space in $R$.

\subsubsection{Data structure}
In this section, $r$ is a certain radius such that $I_r(\cdot)$ and $I_{r^*}(\cdot)$ 
have the same combinatorial structure.
We build a balanced BST(binary search tree) on $S$ as follows.
Let $\bbst{S}$ denote a balanced BST on an ordered point set $S$.
Each leaf node corresponds to an ordered point in $S$ in order, from left to right. 
Each nonleaf node corresponds to the points of $S$ corresponded to by  
the leaf nodes of the subtree rooted at the node.
For a node $w$, let $S_w$ denote the set of points corresponding to $w$. 
See Figure~\ref{fig:canonical}(a).
% and canonical subset for ordered pointset $S$. 
% Each leaf node represents an ordered point in $S$ in order, from left to right. 
% An internal node $w$ represents the union of the points represented by
% its children nodes. 

Observe that any interval $[i,j]$ of points of $S$ can be presented by 
$O(\log n)$ nodes and they can be found in $O(\log n)$ time
as a range can be represented by $O(\log n)$ subtrees 
in the 1-dimensional range tree on $S$.
We call these nodes \emph{the canonical nodes} of the interval in $\bbst{S}$. 
For instance, the internal node corresponding to $\{p_1,p_2\}$ and 
the leaf node corresponding to $p_3$
are the canonical nodes of interval $[1,3]$ in Figure~\ref{fig:canonical}(a).
% \complain{Check!}
%We can create new data structure $\mathcal{C}(P)$ modified from $\mathcal{C}_1(P)$. 
At each nodes $w$ in $\bbst{S}$, we store $I_r(S_w)$ as an additional information.
% \complain{without knowing $r^*$ in advance?} 
% such that \ccheck{$I(S_w)$ and $I_{r^*}(S_w)$} have the same combinational structure. 

The boundary of $I_r(S_w)$ is represented by a balanced binary search tree
and stored at $w$.
The balanced binary search trees for the nodes in $\bbst{S}$ are 
computed in bottom-up manner in $O(n\log n)$ total time. 
See Figure~\ref{fig:canonical}.
% where $S_w$ is the set of points represented by $w$.
Thus, $\bbst{S}$ can be constructed in $O(n\log n)$ time.

\begin{figure}[ht]
	\centering
	\includegraphics[width=.9\textwidth]{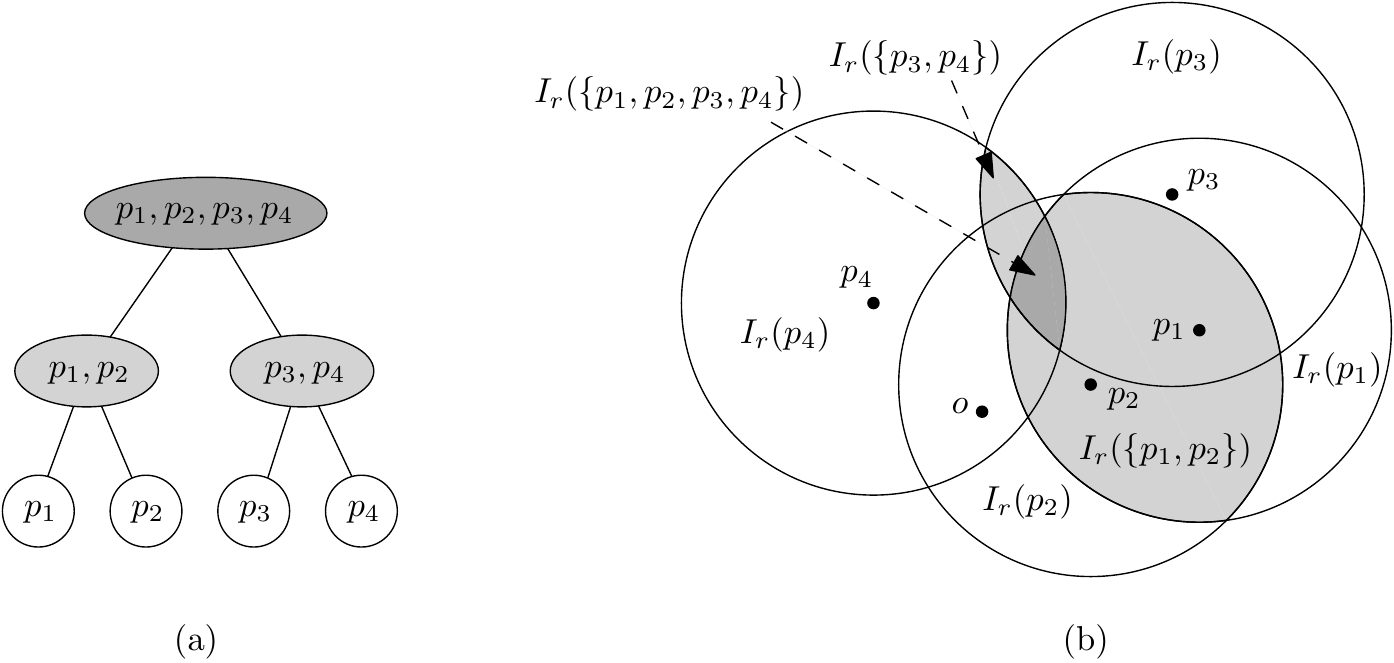}
	\caption{(a) The balanced binary search tree $\bbst{S}$ for $S=\{p_1,p_2,p_3,p_4\}$. (b) The boundary of $I_r(S_w)$ for a node $w$ in the tree in (a) is represented by 
	a balanced binary search tree and stored at $w$.}
	\label{fig:canonical}
\end{figure}

The following two technical lemmas can be shown by lemmas %~$6$ and Lemma~$5$ 
in Wang's paper~\cite{WANG2020} as $I_r(S_w)$ is the dual of the circular hull $\alpha_r(S_w)$.

\begin{lemma}[Lemma~$6$ in~\cite{WANG2020}]
\label{lem:TS}
$\bbst{S}$ can be constructed in $O(n\log n)$ time such that the combinatorial structure of $I_r(S_w)$ is same as $I_{r^*}(S_w)$ for any $w \in \bbst{S}$.
\end{lemma}
%\begin{proof}
%\ccheck{Dual of Wang's data structure.}
%\end{proof}

%Now we will give some technical lemmas for queries on $\bbst{S}$ and 
%operations used in the algorithm. 
%%\ccheck{In the algorithm, we need to determine if two common intersections of disks have
%%a nonempty intersection. This can be done by applying some operations
%%on the balanced binary search trees representing their boundaries.}
%% The basic operation is used for constructing data structure such that finding intersection of disk.} \complain{???} 
%The following lemma can be shown by using Lemma~$5$ in Wang's paper~\cite{WANG2020}
%as $I(S_w)$ is the dual of the circular hull of $S_w$.
\begin{lemma}[Lemma~$5$ in~\cite{WANG2020}]\label{lem:tangent}
  Let $L$ and $R$ be the point sets in the plane such that $L$ and $R$
  are separated by a line and the arcs of $I_r(L)$ and $I_r(R)$ are stored
  in a data structure supporting binary search.  One can
  do the following operation in $O(\log(|L|+|R|))$ time: determine
  whether $I_r(L) \cap I_r(R)=\emptyset$ or not; if
  $I_r(L) \cap I_r(R)\neq\emptyset$, either determine whether
    $I_r(L) \subseteq I_r(R)$ or $I_r(R) \subseteq I_r(L)$, 
  or find the two intersection points of $\partial I_r(L)$ and
  $\partial I_r(R)$.
\end{lemma} 
%\begin{proof}
%\ccheck{Dual of Wang's common tangent.}
%\end{proof}
 
\subsubsection{Speeding up intersection queries}
Given a query range $[i,j]$, we can compute
the common intersection $I(S^+[i,j])$ for points in $S^+[i,j]$
by computing the common intersection of $I(S_w)$'s 
for all canonical nodes $w$ of $S^+[i,j]$ in $\bbst{S^+}$
by Lemma~\ref{lem:tangent}, and by splitting and gluing 
the binary search trees stored in the canonical nodes. 
Wang showed a parallel algorithm for this process
using $O(\log n)$ processors and running in $O(\log n \log \log n)$ time.
Wang's algorithm handles gluing two binary trees at a processor.
% $I_1,\ldots,I_k$ each of which is the intersections of congruent disks.
% centered at points. in $S_w$ of a canonical node $w$.

We can improve the query time for our purpose by using
more processors, and by splitting and gluing 
multiple (more than two) binary trees simultaneously to compute
the boundary arcs of common intersections
using the order of the boundary arcs.
% Wang's algorithm handles gluing two binary trees at a processor.
Our algorithm computes the part of $\partial I_i$ that appears on the boundary
of the common intersection. % and combined the boundary.}
To compute the boundary part of $I_i$ efficiently, we represent the intersection
of two regions in a number of intervals on its boundary. 

\begin{figure}[ht]
	\centering
	\includegraphics[width=.6\textwidth]{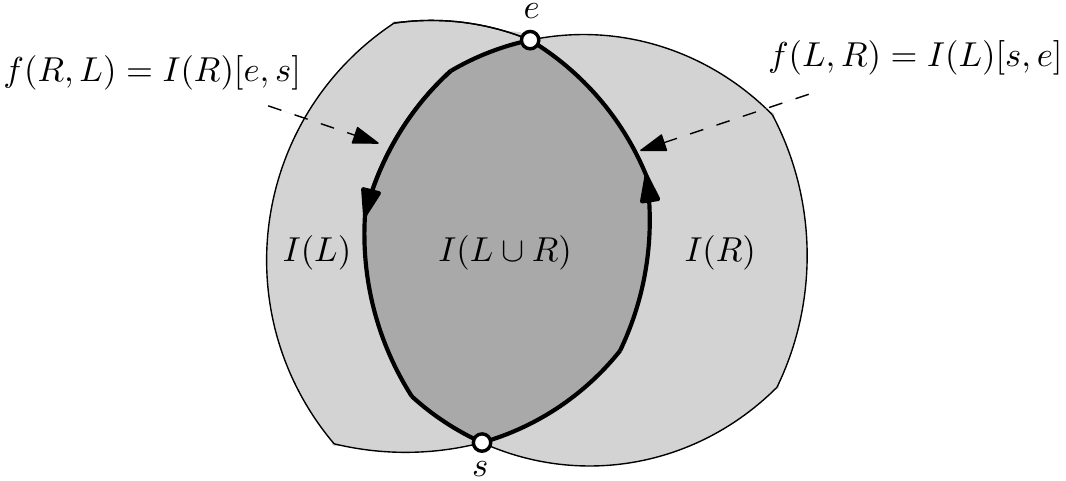}
	\caption{$I(L) \cap I(R)$ consists of two boundary parts, $f(L,R)$ and $f(R,L)$.}
	\label{fig:f_function}
\end{figure}

For a set of points in the plane, 
let $L$ and $R$ be the subsets of the points separated by a line. 
% partitions a set of points in the plane 
% into two subsets $L$ and $R$ such that the points of $L$
% lie in one side of the line and the points of $R$ lie on the other side in the plane.}
Then $\partial I(L)$ and $\partial I(R)$ intersect as most twice.
Thus, we can represent $I(L)\cap I(R)$ by two boundary parts, 
one from $\partial I(L)$ and one from $\partial I(R)$, as follows. 
Let $f(L,R)$ and $f(R,L)$ denote the parts of $\partial I(L)$ and $\partial I(R)$,
respectively, such that $f(L,R)$ and $f(R,L)$ together form the boundary of
$I(L)\cap I(R)$. We can represent $f(L,R)$ by the two intersection points 
$\partial I(L)\cap \partial I(R)$ and the counterclockwise direction along 
$\partial I(L)$.
For two intersection points $s,e$ of $\partial I(L)\cap \partial I(R)$, we use
$f(L,R)=I(L)[s,e]$ if $s$ appears before $e$ on $\partial I(L)\cap \partial I(R)$
along $\partial I(L)$ in counterclockwise direction. See Figure~\ref{fig:f_function}.

For a family $\mathcal{W}$ of point sets, let $I(\mathcal{W})=I(\bigcup_{W\in\mathcal{W}} W)$ 
for ease of use.
The boundary $\partial I(\mathcal{W})$ consists of  
$I(\mathcal{W}) \cap \partial I(W)$ for each subset $W \in \mathcal{W}$.
To compute $I(\mathcal{W}) \cap \partial I(W)$, we need following lemma.
We use $U^+(i,j)$ (and $U^-(i,j)$) to denote the set of the canonical nodes of 
range $[i,j]$ for $\bbst{S^+}$ (and $\bbst{S^-}$). 

\begin{lemma}\label{lem:intersection}
  Let $X$ be a fixed point set and let $\mathcal{W}$ be a family of
  point sets.  If we have $f(X,W)$ for every $W \in \mathcal{W}$ and
  a point $x \in I(X)$, there are at most $|\mathcal{W}|$ connected components of
  $I(\mathcal{W})\cap \partial I(X)$, each of which can be represented by a connected part of $\partial I(X)$,
  and we can compute them in $O(|\mathcal{W}| \log |\mathcal{W}|)$ time.
  If $X=S_w$ for a node $w \in U^+(i,j)$ and $\mathcal{W}=\{S_{w'} \mid w' \in U^+(i,j)\}$, 
  there are at most two such connected components and we can compute them 
  in $O(|\mathcal{W}|)$ time .
%  Let $X$ be a fixed point set and let $\mathcal{W}$ be a family of
%  point sets.  If we have $f(X,W)$ for every $W \in \mathcal{W}$ and
%  a point $x \in I(X)$, we can compute the connected components of
%  $I(\mathcal{W})\cap \partial I(X)$ in $O(|\mathcal{W}|)$ time, each
%  of which can be represented by a connected part 
%  of $\partial I(X)$.
%  There are at most three such connected parts of $\partial I(X)$.
%  There are at most two such intervals 
%  if there is a region $R$ bounded by two rays emanating from the origin going upward 
%  such that $X\subset R$ and $W\cap R=\emptyset$ for every $W\in \mathcal{W}\setminus\{X\}$.
\end{lemma}
\begin{proof}

  For a set $W \in
  \mathcal{W}$, consider $f(X,W)=I(X)[s,e]$.  For a point $x\in I(X)$, 
  $I(X)[s,e]$ can be considered as an angle interval $[\theta_{s},\theta_{e}]$ 
  with respect to $x$, where $\theta_s$ and
  $\theta_e$ are the angles of $\overrightarrow{xs}$ and $\overrightarrow{xe}$. 
  (Note that the endpoints of $I(X)[s,e]$ can be represented as algebraic functions
  of $r$ with constant degree for $r$. 
%  such that $I_r(\cdot)$ and $I_{r^*}(\cdot)$ have the same combinatorial structure. 
  For ease of description, we simply use the angles instead.)
  
%   \ccheck{$I(\mathcal{W})\cap \partial I(X) = \cap_{W \in \mathcal{W}} I(W) \cap \partial I(X) = \cap_{W \in \mathcal{W}} (I(W) \cap \partial I(X))= \cap_{W \in \mathcal{W}} f(X,W)$}

%To compute $I(\mathcal{W})\cap \partial I(X)$, which is $\bigcap_{W \in \mathcal{W}} f(X,W)$,

	By sorting the endpoints of the angle intervals, we can compute 
	$\bigcap_{W \in \mathcal{W}} f(X,W)$, which is $I(\mathcal{W})\cap \partial I(X)$.    
  However, there are only two such angle intervals if $X=S_w$ for a node 
  $w \in U^+(i,j)$ and $\mathcal{W}=\{S_{w'} \mid w' \in U^+(i,j)\}$, 
  and thus we can compute $I(\mathcal{W})\cap \partial I(X)$ in $O(|\mathcal{W}|)$ time
  as follows.
  Let $S_w=S^+[i',j']$, then $\mathcal{W}$ can be partitioned into three subfamilies 
  $\{S_w\}, \mathcal{W}_L=U^+(j'+1,j)$, and $\mathcal{W}_R=U^+(i,i'-1)$. Observe 
  that the subfamilies are separated by the two lines, both through the origin, 
  one through $p_{j'}$ and one through $p_{i'}$.
  Let $A_L$ be the set of angle intervals from $f(X,W_L)$ for $W_L \in \mathcal{W}_L$ and 
  let $I_L$ be the common intersection of the angle intervals in $A_L$. 
  The common intersection of the angle intervals
  of any subset of $A_L$ consists of at most one angle interval by Lemma~\ref{lem:tangent}, because 
  any set of $\mathcal{W}_L$ is separated from $S_w$ by a line. 
  Consider an angle interval $[\theta_s,\theta_e]$ of $A_L$. Then the intersection of 
  $[\theta_s,\theta_e]$ with every other angle interval in $A_L$ is connected. 
  Thus, we compute the common intersection $I_L$ of
  the angle intervals of $A_L$ within $[\theta_s,\theta_e]$ in $O(|A_L|)$ time.
For the set of angle intervals from $f(X,W_R)$ for $W_R \in \mathcal{W}_R$,
we can compute the common intersection of those intervals in a similar way.
  Thus there are at most two such connected components and we can compute 
  them in $O(|\mathcal{W}|)$ time.
\end{proof}

In the following, we abuse $I(i,j)$ to denote $I(S^+[i,j])$ for any two indices $i,j$ satisfying 
$1\leq i\leq j\leq n$.
To merge the boundaries we need to know their appearing order in boundary of common intersection.

\begin{lemma}\label{lem:boundary}
The indices of the points corresponding to the arcs of $\partial I(i,j)$ are increasing and decreasing consecutively 
at most once while traversing along $\partial I(i,j)$.
\end{lemma}
\begin{proof} 
We use circular hulls $\alpha_r(\cdot)$ to prove the lemma.
Since $I_r(\cdot)$ and $\alpha_r(\cdot)$ are dual to each other,
the proof also holds for $I_r(\cdot)$.
We simply use $\alpha(\cdot)$ to denote $\alpha_r(\cdot)$. 
% such that
% $\alpha_r(\cdot)$ and $\alpha_{r^*}(\cdot)$ have the same combinatorial structure.
For any two indices $i$ and $j$, we show that 
the indices of the vertices that appear on the boundary of $\alpha(S^+[i,j])$ are increasing and decreasing consecutively at most once.

Consider the case that $\alpha(S^+[i,j])$ does not contain $o$.
Let $x_1$ and $x_2$ be the contact points 
of the right and left tangent lines to $\alpha(S^+[i,j])$ from $o$. 
For a point $p$ on the boundary of $\alpha(S^+[i,j])$, the angle between $\overrightarrow{op}$ and the $x$-axis 
is increasing while moving along the boundary of $\alpha(S^+[i,j])$ from $x_1$ to $x_2$ in counterclockwise order. 
Similarly, the angle between $\overrightarrow{op}$ and the $x$-axis is decreasing while moving 
along the boundary of $\alpha(S^+[i,j])$ from $x_2$ to $x_1$ in counterclockwise order.
Thus, the indices of the vertices of the circular hull are increasing and decreasing at most once. 

Now consider the case that $\alpha(S^+[i,j])$ contains $o$.
By the definitions of $o$ and $S^+$, there is no vertex of $\alpha(S^+[i,j])$ lying below the $x$-axis. 
Thus, for a vertex $v$ of $\alpha(S^+[i,j])$, the angle between $\overrightarrow{ov}$ and the $x$-axis is 
increasing while moving along the boundary of $\alpha(S^+[i,j])$ from $p_s$ to $p_t$ in counterclockwise, 
where $s$ and $t$ are the smallest and largest indices of the vertices of $\alpha(S^+[i,j])$. 

Since the order of the vertices on the boundary $\alpha(S^+[i,j])$ is the same as the order of the indices of points 
corresponding to the arcs of $\partial I(i,j)$, the lemma holds.
\end{proof}

We construct $\bbst{S^+}$ (and $\bbst{S^-}$) in $O(n\log n)$ time using Lemma~\ref{lem:TS}.
Now, we have basic lemmas on queries to $\bbst{S^+}$  (and $\bbst{S^-}$).  
% \ccheck{We call the sets of points represented by the nodes \emph{the canonical subsets} $U(i,j)$.} \complain{Check!}
% See figure~\ref{fig:canonical}.
Let $(r_1,r_2]$ be the range of radius $r$ for which $I_r(S_w)$ and $I_{r^*}(S_w)$ have the same
combinatorial structure for any $w\in\bbst{S^+}$.
\begin{lemma}\label{lem:query}
  Once $\bbst{S^+}$ is constructed, 
  we can process the following queries in $O(\log n)$ time 
  using $O(\log^2 n)$
  processors: given $r \in (r_1,r_2]$ and any pair $(i,j)$ of indices, 
  determine whether $I(i,j)=\emptyset$ or not, and if
  $I(i,j)\neq\emptyset$, return the root of a balanced binary search
  tree representing $I(i,j)$. 
\end{lemma}
\begin{proof}
%  We will return a binary search tree of $\partial I[i,j]$ whose height
%  is $O(\log n)$.  
  Observe that $\partial I(i,j)$ consists of parts of
  $\partial I(S_w)$ for all canonical nodes $w \in U^+(i,j)$. For a fixed
  $w \in U^+(i,j)$, we compute $I(i,j) \cap \partial I(S_w)$ by
  taking $f(S_w,S_{w'})$ for every canonical node $w' \in U^+(i,j)$ and applying
  Lemma~\ref{lem:intersection}. 
  %We find them in $O(\log n)$ time.} 
  For a fixed $w\in U^+(i,j)$, we can compute $f(S_w,S_{w'})$ for every $w'\in U^+(i,j)$
  in $O(\log n)$ parallel steps using $O(\log n)$ processors. 
 % Since there is a region $R$ bounded by two rays emanating from the origin 
%  going upward such that $S_w\subset R$ and $S_{w'}\cap R=\emptyset$ for every $S_{w'}$ with $w'\in U^+(i,j)\setminus\{w\}$,
%  If $X=S_w$ for a node $w \in U^+(i,j)$ and $\mathcal{W}=\{S_{w'} \mid w' \in U^+(i,j)\}$, 
  By Lemma~\ref{lem:intersection},
  there are at most two boundary parts representing $I(i,j) \cap \partial I(S_w)$, and we can 
  find them in $O(\log n)$ time using one processor.
  
  For a part $I(S_w)[s,e]$, we construct a binary search tree for
  $I(S_w)[s,e]$ using $I(S_w)$ stored at $w \in \bbst{S^+}$
  and the path copying~\cite{drisco1989} in $O(\log n)$ time.
  Since there are $O(\log n)$ nodes in $U^+(i,j)$, we can find
  $O(\log n)$ binary search trees, at most two for each node,
  in $O(\log n)$ parallel steps using $O(\log^2 n)$ processors. Observe that the
  binary search trees for a point set $S_w$ represents $I(i,j) \cap \partial I(S_w)$.
%    and
%  the union of the $O(\log n)$ binary search trees represents $\partial I[i,j]$.

  Then we merge all these binary search trees % computed in the \ccheck{previous step (paragraph)}
  into a binary search tree representing $\partial I(i,j)$.  By
  Lemma~\ref{lem:boundary}, the indices of the points in $S^+[i,j]$ corresponding to the arcs of 
  $\partial I(i,j)$ are increasing and then decreasing consecutively at most once while traversing along $\partial I(i,j)$.
For any two distinct canonical nodes $w,w'$ of $U^+(i,j)$, we use $w<w'$ 
if the indices of the points in $S_w$ are smaller than the indices of the points in $S_{w'}$.
Since the canonical nodes of $U^+(i,j)$ can be ordered by their corresponding point sets, 
the binary trees can also be ordered accordingly, 
by following the orders of their corresponding canonical nodes. %\complain{Pairs are ordered!}
Each canonical node $w$ has at most two parts of $I(i,j) \cap \partial I(S_w)$,
and the order between the binary search trees on the parts is decided by the order 
of their corresponding canonical nodes
in $\bbst{S^+}$.
%  \ccheck{the order of the binary search
%  trees along $\partial I[i,j]$ is known,
%  the indices increase and decrease at most once, and each index
%  has only two intervals.} \complain{???}  

Following this order, we merge the $O(\log n)$ binary search trees 
in $O(\log n)$ time. 
Recall that there are at most two binary search trees $T, T'$ for each node $U^+(i,j)$,
which are at the same position in the order. Thus, 
the merge process is done in two passes, one in the order and one in
reverse of the order, and $T$ is merged as a boundary part of $I(i,j)$ in one pass 
and $T'$ is merged as a boundary part of $I(i,j)$ the other pass.
% \ccheck{While the indices consisting of $\partial I(i,j)$ are increasing, the boundary parts
% of $\partial I(i,j)$ are from the binary search trees, possibly skipping some of them, in order.
% This can be be done by comparing the endpoints of the boundary parts represented by
% the trees.}
% \complain{Details on skipping?}
% We have at most two such candidate trees at a time. 
By repeating this process on the binary search trees in order, 
we can construct $\partial I(i,j)$. Finally we
  can return the root of the merged binary search tree representing
  $I(i,j)$ in $O(\log n)$ time using $O(\log^2 n)$ processors.
\end{proof}

\subsubsection{Reducing the search space to submatrices}
Our algorithm reduces the set of candidates for $r^*$ from the elements 
in the $n\times n$ matrix $R=(r_{ij}$) to the elements in $O(n/\log^6 n)$ 
disjoint submatrices of size $\log^6 n\times \log^6 n$ in $R$. 
This is done by evaluating on $r_{ij}$ at every
$\log^6 n$-th $j$ values from $j=1$, which results in a set of $O(n/\log^6 n)$
disjoint submatrices of of $R$ with height $\log^6 n$.
Then we divide each of these submatrices further such that its width is at most $\log^6 n$. 
See Figure~\ref{fig:monotone}.
% So there are $O(n/\log^6 n)$ groups.}
% For the reducing candidates we do as follows.

Precisely, let $m=\lfloor n / \log^6 n \rfloor$, $j_t= t \times \lfloor n/m \rfloor$ for $t = 0,1, \ldots, m$. 
%\JMremove{Let $j_0=0$ and $j_m=n$.} 
For each $t \in [0,m]$, let $i_t$ be the largest index in $[0,n]$
satisfying $A[i_t,j_t] \geq B[i_t,j_t]$. Observe that
$i_0 \leq i_1\leq \ldots \leq i_m$. Each $i_t$ can be found in
$O(\log^7 n)$ time after $O(n\log n)$-time preprocessing.
Since there are $O(m)$ such $i_t$'s, we can compute them in 
$O(n\log n)$ time~\cite{CHAN1999189}. % \complain{$O(n\log^7 n)$ time???}

The algorithm determines whether $r_i^* \leq r$ or not, for all $i =0,1,\ldots, n$, as follows.
For each $i$, let $t$ be an index in $[0,m-1]$ satisfying $i_t < i \leq i_{t+1}$.  
If $A[i,j_t] > r$, then the algorithm returns $r^*_i> r$. Otherwise, 
it finds the largest index $j \in [j_t,j_{t+1}]$ satisfying $A[i,j] \leq r$, and  returns $r_i^* \leq r$ if and only if $B[i,j] \leq r$.  See Algorithm 2 of~\cite{WANG2020} and Theorem 4.2
of~\cite{CHAN1999189}.

Our algorithm divides the indices $i$ from 0 to $n$ into at most $2m$ groups.
For each $t=0,1,\ldots,m-1$, if $i_{t+1}-i_t \leq \log ^6 n$ , the algorithm forms a group of at most $\log^6 n$ indices. Otherwise it forms a group for every consecutive $\log^6 n$ indices up to $i_{t+1}$.  Then there are at most $2m$ groups. Each group $G=[a,b]$ is contained in one of $[i_t,i_{t+1}]$. We identify group $G$ as $G(a,b,t)$. See Figure~\ref{fig:monotone}. 
\begin{figure}[ht]
	\centering
	\includegraphics[width=\textwidth]{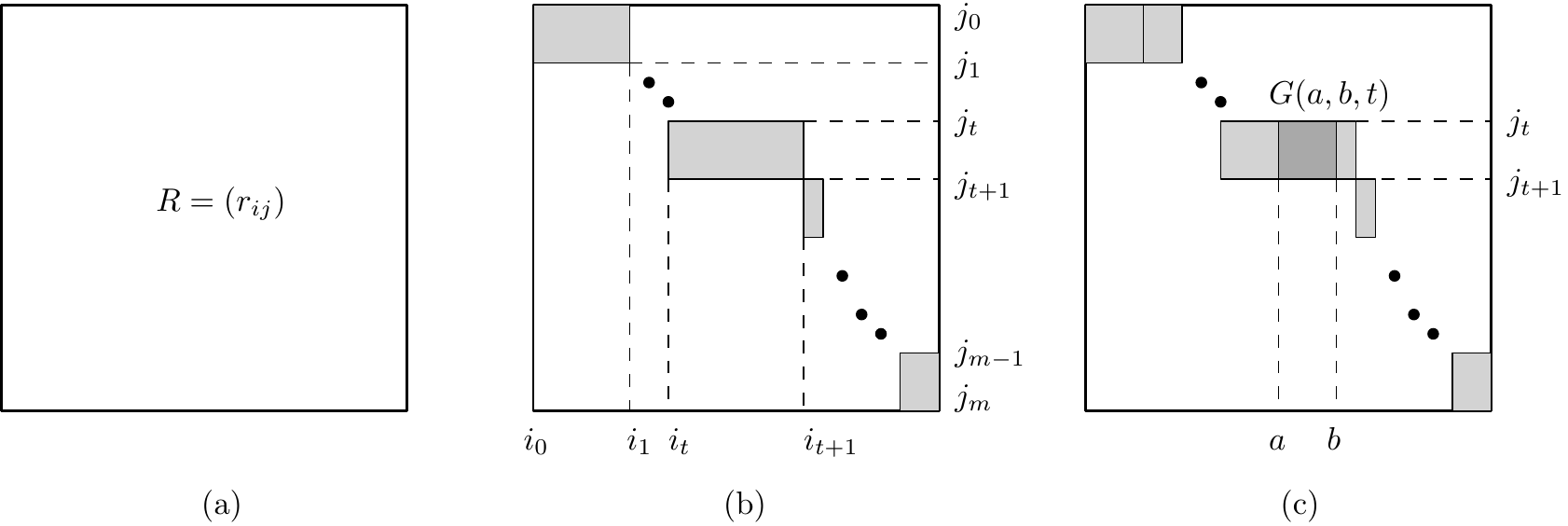}
	\caption{(a) Matrix $R=(r_{ij})$. (b) Reducing the search space into $O(n/\log^6 n)$ 
	submatrices of height $\log^6 n$ using the 
	monotone property, at every $(\log^6 n)$-th entries of index $j$. (c) Dividing each submatrix further
	into one with  width at most $\log^6 n$. A submatrix (a group of indices) $G(a,b,t)$ represents the 
	index ranges of $i$ from $a$ to $b$ and of $j$ from $j_t$ to $j_{t+1}$.}
	\label{fig:monotone}
\end{figure}

%At preprocessing phase, we construct $\bbst{S^+}$ and $\bbst{S^-}$, reduce candidates by finding $i_t$ and divides indices to $O(m)$ groups. 
For each group $G(a,b,t)$, we construct $\bbst{S^+[a,b]}$ and 
$\bbst{S^-[j_t,j_{t+1}]}$, which will be used to construct a data structure for determining 
the emptiness of intersections.
% This data structure will be used for constructing determining data structure.
This can be done for all groups in $O(n\log n)$ time in total.

\subsection{Phase 2 - Group information}
Given a value $r>0$, we determine if $A[i,j] \leq r$ for a group $G(a,b,t)$
with $a \leq i \leq b$ and $j_t \leq j \leq j_{t+1}$. % Consider we determine if $A[i,j] \leq r$ for a group $G(a,b,t)$
% with $a \leq i \leq b$ and $j_t \leq j \leq j_{t+1}$.
Observe that the points of $S^+[b,n]$ and the points of $S^-[1,j_t]$ are used
commonly in the process. See Figure~\ref{fig:dividingsets}.
\begin{figure}[ht]
	\centering
	\includegraphics[width=.35\textwidth]{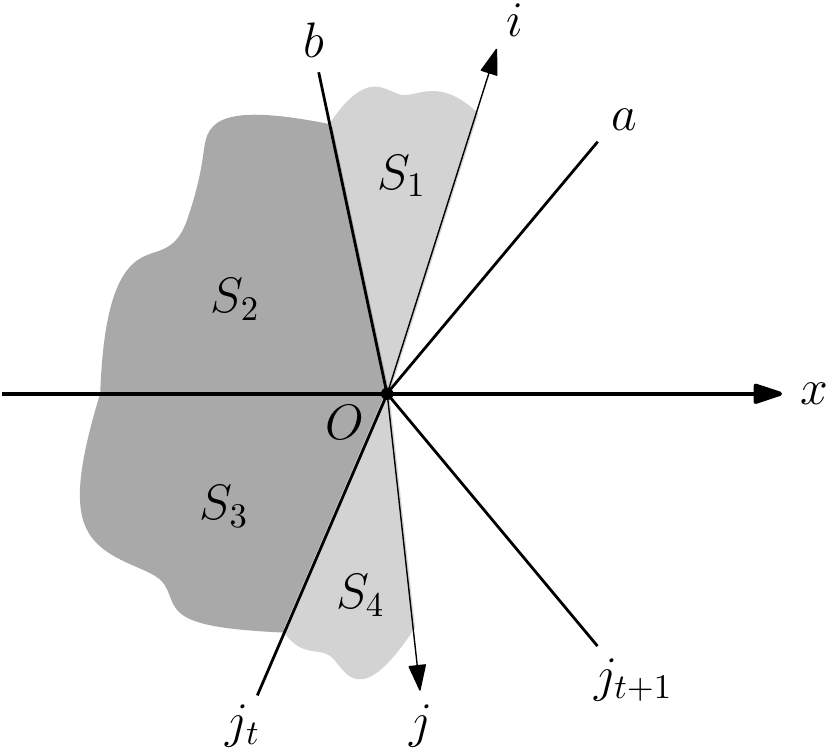}
	\caption{To determine if $A[i,j]\leq r$ for a $G(a,b,t)$ with $a \leq i \leq b$, $j_t \leq j \leq j_{t+1}$,
	we use  four subsets $S_1=S^+[i,b], S_2=S^+[b,n], S_3=S^-[1,j_t],$ and $S_4=S^-[j_t,j]$. 
	Observe that $b,j_t$ and $n$ are constant, and $|b-a|,|j_{t+1}-j_t|\leq \log^6 n$.}
	\label{fig:dividingsets}
\end{figure}

%\complain{change term of preprocessing} 
The decision algorithm consists of 
computing group information and binary search.  
We simply use $I(\cdot)$ to denote $I_r(\cdot)$.
For for each group $G(a,b,t)$, the algorithm computes 
 $I(S^+[b,n])$ and $I(S^-[1,j_t])$. And then it
constructs a data structure such that given query indices $i$ and $j$ with $a \leq i \leq b$ and $j_t \leq j \leq j_{t+1}$, it determines if $I(S^+[i,n]) \cap I(S^-[1,j]) =\emptyset$ or not
using $I(S^+[b,n])$ and $I(S^-[1,j_t])$. % \complain{$a$ is not used?}
In the binary search, we determine $A[i,j] \leq r$ or not (and $B[i,j] \leq r$ or not) using 
emptiness queries on this data structure. 

For an ordered point set $L$ and any point set $R$ 
separated by a line in the plane, 
let $\mathcal{D}(L,R)$ be a balanced BST on the ordered set $L$ with respect to $R$
constructed from $\bbst{L}$ such that for each node $w$ in $\bbst{L}$, 
we store $f(S_w,R)$ and $f(R,S_w)$ instead of $I(S_w)$.
%\complain{each node $w$ in $\mathcal{D}(L,R)$??} 
% We construct a data structure from 
% \ccheck{Let $\mathcal{D}(L,R)$ be modified data structure from $\bbst{L}$, where $L$ and $R$ are 
% line separated continuous ordered point set.} \complain{???}
%If we change $I(S_w)$ to $f(S_w,R)$ and $f(R,S_w)$ for each nodes $w$ in $\bbst{L}$, the data structure become $\mathcal{D}(L,R)$. 
This data structure will be used to determine if $I(L' \cup R)=\emptyset$ or not for 
any interval $L' \subseteq L$.

\begin{lemma}\label{lem:data3}
$\mathcal{D}(L,R)$ can be constructed in $O(\log (|L|+ |R|))$ time using $O(|L|)$ processors 
once we have $\bbst{L}$ and $I(R)$, where $L$ and $R$ are line separated continuous ordered point set.
\end{lemma}
\begin{proof}
First we assign a processor to a node $w \in \bbst{L}$. 
% $\bbst{L}$ is modified from BBST using $P$. 
Since there are $O(|L|)$ nodes in $\bbst{L}$ and  the tree depth is $O(\log |L|)$, 
we can assign $O(|L|)$ processors in $O(\log |L|)$ time. For each nodes $w$, we compute 
$f(S_w,R)$ and $f(R,S_w)$. This can be done using $I(S_w), I(R),$ and Lemma~\ref{lem:tangent} in $O(\log (|S_w| +|R|)$ time. 
Since $S_w \subseteq L$, $\mathcal{D}(L,R)$ can be constructed in $O(\log (|L| +|R|))$ time
using $O(|L|)$ processors.
\end{proof}

We will apply binary search for each index $i$. Each index $i$ is contained in one of the $2m$ groups. Each group $G(a,b,t)$ consists of continuous $O(\log^6 n)$ indices $\lbrace a,a+1, \ldots, a+b \rbrace$. 
For an index $i \in G(a,b,t)$, the decision on $A[i,j] \leq r$ is equivalent to the decision on
$I_{ij} = \bigcap_{k=1}^4 I(S_k)\neq\emptyset$, %\cap I(S_2) \cap I(S_3) \cap I(S_4)$,} 
where $S_1=S^+[i,b], S_2=S^+[b,n], S_3=S^-[1,j_t],$ and $S_4=S^-[j_t,j]$. 
The emptiness of $I_{ij}$ is equivalent to the emptiness of $\partial I_{ij}$. 
Since $\partial I_{ij}$ consists of boundary parts of 
$I(S_k)$ for each $k=1,\ldots, 4$, the emptiness of $I_{ij}$ is equivalent to the 
emptiness of $I_{ij} \cap \partial I(S_k)$ for every $k=1,\ldots, 4$. %\complain{every $i$?}
% $I_{ij} \cap \partial I(S_2)$, $I_{ij} \cap \partial I(S_3)$, and $I_{ij} \cap \partial I(S_4)$.}} \complain{??? I do not understand what this means.} 
%See Figure~\ref{fig:dividingsets} and Figure~\ref{fig:dividingbound}.

%\begin{figure}[tb]
%	\centering
%	\includegraphics[scale=1]{Boundary-divide.pdf}
%	\caption{Note that $\partial I_{ij}$ consist by $I_{ij} \cap \partial I(S_1)$, $I_{ij} \cap \partial I(S_2)$, $I_{ij} \cap \partial I(S_3)$, and $I_{ij} \cap \partial I(S_4)$}
%	\label{fig:dividingbound}
%\end{figure}

Now we list all the information we need in computing $\partial I(S_1)$ and $I_{ij} \cap \partial I(S_2)$.
Observe that $I_{ij} \cap \partial I(S_1)$ consists of parts of the boundaries $I_{ij} \cap \partial I(S_w)$ for the
canonical nodes $w \in U^+(i,b)$.
To compute $I_{ij} \cap \partial I(S_w)$ for a node $w \in U^+(i,b)$ using Lemma~\ref{lem:intersection}, 
we need $f(S_w,S_1)$, $f(S_w,S_2)$, $f(S_w,S_3)$ and $f(S_w,S_4)$. 
We compute $f(S_w,S_2)$ and $f(S_w,S_3)$ in this Phase. 
Then we compute $f(S_w,S_1)$ and $f(S_w,S_4)$ in Phase 3.

To compute $I_{ij} \cap \partial I(S_2)$ using lemma~\ref{lem:intersection}, 
we need $f(S_2,S_3)$ and $f(S_2,S_w)$ for every node $w \in U^+(i,b) \cup U^-(j_t,j)$.
We compute them all in this Phase.
We can compute $I_{ij} \cap \partial I(S_4)$ and $I_{ij} \cap \partial I(S_3)$ in a similar way. % to $I_{ij} \cap I(S_1)$ and $I_{ij} \cap I(S_2)$.

To cover the query range, we need to compute $f(S_w,S_2)$, $f(S_w,S_3)$, 
$f(S_2,S_w)$ and $f(S_3,S_w)$ for every 
$w \in \bbst{S^+[a,b]} \cup \bbst{S^-[j_t,j_{t+1}]}$
because $i\in[a,b]$ and $j\in [j_t,j_{t+1}]$ while applying the binary search on the group.

We compute the followings for each group $G(a,b,t)$ in this phase.

\begin{enumerate}
\item \label{item:CommonGroup} $I(S_2)$ and $I(S_3)$.
\item \label{item:CommonGroup2} $f(S_2,S_3)$ and $f(S_3,S_2)$.
\item \label{item:treeInformation1} $\mathcal{D}(S^+[a,b],S_2)$, $\mathcal{D}(S^+[a,b],S_3)$, $\mathcal{D}(S^-[j_t,j_{t+1}],S_2)$ and $\mathcal{D}(S^-[j_t,j_{t+1}],S_3)$.
\end{enumerate}

Part~\ref{item:CommonGroup} can be done in $O(\log n)$ time using
$O(\log^2 n)$ processors by Lemma~\ref{lem:query}.
Part~\ref{item:CommonGroup2} can be done in $O(\log n)$ time by
Lemma~\ref{lem:tangent}.
Part~\ref{item:treeInformation1} can be done $O(\log n)$ time using $O(\log^6 n)$ processors by Lemma~\ref{lem:data3}.
Since there are $2m= O(n/\log^6 n)$ such groups, the three parts can be done $O(\log n)$ 
time with $O(n)$ processors in total.

\subsection{Phase 3 - Binary search}

For an index $i \in G(a,b,t)$, we apply binary search over range $j \in [j_t,j_{t+1}]$. 
Since $|j_{t+1}-j_t|=O(\log^6 n)$, our algorithm performs $O(\log\log n)$ steps of 
binary search to determine whether $r_i^* \leq r$ or not.

To determine $I_{ij}\neq\emptyset$, we determine $I_{ij} \cap \partial I(S_1)\neq\emptyset$, $I_{ij} \cap \partial I(S_2)\neq\emptyset$, $I_{ij} \cap \partial I(S_3)\neq\emptyset$, or $I_{ij} \cap \partial I(S_4)\neq\emptyset$.
To determine $I_{ij} \cap \partial I(S_1) \neq \emptyset$, we determine $I_{ij} \cap \partial I(S_w) \neq \emptyset$ 
for every canonical node $w \in U^+(i,b)$. 
To determine $I_{ij} \cap \partial I(S_w)\neq\emptyset$ by Lemma~\ref{lem:intersection}, 
we  compute $f(S_w,S_2)$, $f(S_w, S_3),$ and $f(S_w, S_{w'})$ for all $ w' \in U^+(i,b) \cup U^-(j_{t},j)$, 
because the union of the intersections for $S_2$, $S_3$ and all $S_{w'}$ is $I_{ij}$.
We can compute $f(S_w, S_2)$ and $f(S_w, S_3)$ in $\mathcal{D}(S^+[a,b], S_2)$ and 
$\mathcal{D}(S^+[a,b], S_3)$, respectively, at the corresponding nodes $w$.
We can compute $f(S_w, S_{w'})$ using $I(S_w)$ and $I(S_{w'})$, and 
the intersections $I(S_w)$ and $I(S_{w'})$ can be found in 
$\bbst{S^+[a,b]}$ or $\bbst{S^-[j_t,j_{t+1}]}$. 
Since $|S_w|=O(\log^6 n)$ and $|S_{w'}|=O(\log^6 n)$, it takes $O(\log \log n)$ time to compute 
$f(S_w,S_{w'})$ for fixed $w$ and $w'$ by Lemma~\ref{lem:intersection}.
Since $|U^+(i,b) \cup U^-(j_{t},j)|=O(\log \log n)$, it takes $O(\log^2 \log n)$ time to compute 
$f(S_w,S_{w'})$ for all $w' \in |U^+(i,b) \cup U^-(j_{t},j)|$. 
Thus we can determine $I_{ij} \cap \partial I(S_w)\neq\emptyset$ in $O(\log^2 \log n)$ time.  
Since $|U^+(i,b)|= O(\log \log n)$, 
 $I_{ij} \cap \partial I(S_1)$ can be determined $O(\log^3 \log n)$ time with one processor.

To determine $I_{ij} \cap \partial I(S_2)\neq\emptyset$, 
we need $f(S_2,S_3)$ and $f(S_2,S_w)$ for all $w \in U^+(i,b) \cup U^-(j_t,j)$.
We already have $f(S_2,S_3)$ and $f(S_2,S_w)$ can be computed from $\mathcal{D}(S^+[a,b],S_2)$ and $\mathcal{D}(S^-[j_t,j_{t+1}],S_2)$ in $O(\log \log n)$ time. 
So we can determine $I_{ij} \cap \partial I(S_2)\neq\emptyset$ in $O(\log \log n\log\log\log n)$ time with one processor.

Therefore, for a fixed $i$ and a given $j$, we can determine whether
$I_{ij}$ is empty or not in $O(\log^3 \log n)$ time. Remind that the search range for an index 
is $O(\log^6 n)$. There are $O(\log \log n)$ steps of binary search and there are $n+1$ indices
for $i$. The decision step can be done in $O(\log^4 \log n)$ time using $O(n)$ processors.

\begin{theorem}
  The decision problem for the restricted 2-center problem can be
  solved in $O(\log n)$ time using $O(n)$ processors after
  $O(n\log n)$-time preprocessing.
\end{theorem}

\begin{theorem}
  The restricted 2-center problem can be solved in $O(n\log n)$ time.
\end{theorem}
\begin{proof}
We use Cole's parametric search technique~\cite{Cole-parametric} to compute the optimal radius $r^*$. 
To apply the technique, the parallel algorithm must satisfy a bounded fan-in/bounded fan-out requirement.
Instead of analyzing the parallel algorithm directly, we divide the algorithm to 
a constant number of parts such that each part satisfies a bounded fan-in/bounded fan-out. 

  Our parallel decision algorithm computes some information for each group and
  applies binary search on indices.
  In the group information phase, 
  we compute $I(S_2)$ by computing $I(S_2) \cap \partial I(S_w)$ 
  for every $w \in U^+(b,n)$ and merging them. 
  Each $I(S_2) \cap \partial I(S_w)$ is computed by computing $f(S_w,S_{w'})$ for every 
  $w' \in U^+(b,n)$ and merging them. 
  Thus, a processor activates at most one processor, and the fan-out is bounded. 
  We also compute $f(S_2,S_3)$ for each group independently by a processor, 
  and thus the fan-in and fan-out are bounded.
  We assign processors for all nodes $w$ without knowing the decision parameter $r$ in advance.
  Then each processor computes $f(S_2,S_w)$ and $f(S_w,S_2)$ for a node $w$ independently,
  and thus the fan-in and fan-out are bounded. 
  In the binary search phase, each binary search for an index is performed by a processor 
  independently, and thus the fan-in and fan-out are bounded.
  
  Therefore, our parallel algorithm consists of a constant number of fan-in or fan-out
  bounded networks. Thus we can apply Cole's parametric
  search technique to compute $r^*$ in $O((T_S+Q)(T_P+ \log Q))$ time, 
  where $T_S$ is the sequential decision time, $T_P$ is the parallel decision time and 
  $Q$ is the number of processors for parallel decision algorithm. 
  Since we have $T_S=O(n)$, $T_P=O(\log n)$ and $Q=O(n)$,
  $r^*$ can be computed in $O(n\log n)$ time.
\end{proof}

\section{Optimal 2-center for points in convex position}
In this section, we consider the 2-center problem for points in convex position.
Wang gave an $O(n\log n\log\log n)$-time algorithm for
this problem.
% We show how to apply our $O(n\log n)$-time algorithm to the case that
% the points of $S$ are in convex position. 
Let $S$ be a point set consisting of $n$ points in convex position in the plane.
We denote by $\conv(S)$ the convex hull of $S$.
It is known that there is an
optimal solution $(D_1^*, D_2^*)$ such that $D_1^*$ covers a set of
consecutive vertices (points of $S$) along $\partial \conv(S)$ and
$D_2^*$ covers the remaining points of $S$~\cite{kim2000}.  Since the
points are in convex position, for any point $q$ contained in
$\conv(S)$, the points of $S$ appears in the same order around $q$.

Let $\mu_1$ and $\mu_2$ be the two rays from $o$ that separate the
point set $S$ into two subsets, one covered by $D_1^*$ and the other
covered by $D_2^*$.  We need to find a line $\ell$ that separates
$\mu_1$ and $\mu_2$. The line $\ell$ partitions the points of $S$ to $S^+$ and
$S^-$. Then $o$ can be any point in $\ell$.

Wang gave an algorithm that finds the optimal two disks or the line
$\ell$ in $O(n\log n)$ time~\cite{WANG2020}.  The algorithm first
sorts the points of $S$ along the boundary of $\conv(S)$ and pick any
point $p_1\in S$. Then it finds $p^*\in S$ such that the two congruent
smallest disks $(D_1, D_2)$ with $D_1$ covering $\conv(S)[p_1,p^*]$
and $D_2$ covering $S\setminus \conv(S)[p_1,p^*]$ have the minimum
radius over all $p\in S$. % For a fixed $p_1$,
The radius of $D_1$ covering $\conv(S)[p_1,p_j]$ does not decrease
while $p_j$ moves along the boundary of $\conv(S)$.  Similarly, $D_2$
has this property. In each step of binary search, we compute $D_1$
and $D_2$ in $O(n)$ time. Thus, the algorithm can find $p^*$ in
$O(n\log n)$ time using binary search. % and $D_1$ in $O(n)$ time.  
If $\mu_1$ passes through $p_1$, we already have the minimum radius.
Otherwise, $\mu_1$ or $\mu_2$ passes through $p^*$ or one of its two
neighboring point along $\conv(S)$, or $\mu_1$ crosses
$\conv(S)[p_1,p^*]$ and $\mu_2$ crosses $\conv(S)[p^*,p_1]$.  So we
can find the optimal disks or the line $\ell$ that separates $\mu_1$
and $\mu_2$ in $O(n\log n)$ time.

Therefore, we can apply our algorithm in Section~\ref{sec:ours} to the
points in convex position. From this, we improve the running time
by a $\log\log n$ factor over the $O(n\log n\log\log n)$-time algorithm
by Wang.
\begin{theorem}
  The 2-center problem for $n$ points in convex position in the plane
  can be solved in $O(n\log n)$ time.
\end{theorem}
\section{Conclusions}
We presented a deterministic $O(n \log n)$-time algorithm for the case that the centers of 
the two optimal disks are close together, that is, the overlap of the two optimal disks 
is a constant fraction of the disk area. 
Now for the planar 2-center problem, the bottleneck of the time bound
is the case that the two optimal disks are disjoint, and 
Eppstein's $O(n\log^2 n)$-time algorithm is best for the case. 
Thus, the time for the planar 2-center problem still
remains to be $O(n\log^2 n)$ due to the well-separated case.

We also presented a deterministic $O(n\log n)$-time algorithm for 
$n$ points in convex position in the plane. This closes the long-standing
question for the convex-position case.

\bibliography{references}{}

\end{document}